\documentclass{amsart}

\usepackage{amsaddr}

\newtheorem{thm}{Theorem}[section]

\theoremstyle{definition}
\newtheorem{definition}[thm]{Definition}

\theoremstyle{remark}

\numberwithin{equation}{section}

\begin{document}


\title{Computing in the Limit}

\author{Antony Van der Mude}
\address{32 Second Avenue \#332, Burlington MA, 01803}
\address{1-908-343-1334}
\address{vandermude@acm.org}

\begin{abstract}
We define a class of functions termed ``Computable in the Limit'', based on the Machine Learning paradigm of ``Identification in the Limit''. A function is Computable in the Limit if it defines a property $P_p$ of a recursively enumerable class $A$ of recursively enumerable data sequences $\vec{S}\in A$, such that each data sequence $\vec{S}$ is generated by a total recursive function $\phi_s$ that enumerates $\vec{S}$. Let the index $s$ represent the data sequence $\vec{S}$. The property $P_p(s)=x$ is computed by a partial recursive function $\phi_p(s,t)$ such that there exists a $u$ where $\phi_p(s,u)=x$ and for all $t\geq u$, $\phi_p(s,t)=x$ if it converges. Since the index $s$ is known, this is not an identification problem - instead it is computing a common property of the sequences in $A$. We give a Normal Form Theorem for properties that are Computable in the Limit, similar to Kleene's Normal Form Theorem. We also give some examples of sets that are Computable in the Limit, and derive some properties of Canonical and Complexity Bound Enumerations of classes of total functions, and show that no full enumeration of all indices of Turing machines $TM_i$ that compute a given total function $f(x)$ can be Computable in the Limit.

\textbf{Index Terms: Learning, mu Theorem, Computable}
\end{abstract}

\maketitle

\section{Introduction}

Computing in the Limit started with the pioneering work in Machine Learning of Ray Solomonoff\cite{solomonoff1964formal1}\cite{solomonoff1964formal2} and Mark Gold\cite{gold1967language}. They developed some initial results on a type of Machine Learning called Identification in the Limit. Identification in the Limit is a learning task where the learner is presented with a sequence of elements from a recursively enumerable set and is given the task to determine an index $i$ for the function $\phi_i$ whose domain is the set. Besides its applications in learning, Gold\cite{gold1965limiting} considered the problem of identification of sets in the limit (Limiting Recursive Sets), and Kolmogorov based his complexity theory on work by Ray Solomonoff.

One basic result from Kolmogorov Complexity is the function $K(x)=y$, where Turing machine $TM_y$ is the machine with the shortest index $y$ that, starting with an empty tape, halts with $x$ on its tape. The Incompressible Numbers are numbers where, essentially, $K(x)=|x|+c$. This happens when there is no shorter way of generating these numbers except by storing a copy of the number $x$ in the state table of the Turing machine, which is just the size of $x$ plus a constant $c$. These values are known to be immune - no infinite recursively enumerable set is a subset of the Incompressible Numbers.

In this paper, we will consider a more general problem. We cosider the class of functions$f_i(x)$ based on all two-input recursive functions $\phi_i(x,y)$ where $f_i(x)=z$ if $\phi_i(x,y)$ converges to $z$ for increasing values of $y$, and try to discover some properties of of the functions in this class. We will show that the class of Computable in the Limit functions has a nice normal form theorem to characterize the classs, similar to the Kleene Normal Form Theorem. 

The authors Gold\cite{gold1965limiting}, Stephan and Zeugmann\cite{stephan1999uniform} and Terwijn\cite{terwijn2002learning} have studied this class of functions. Gold shows that functions in this class that are are limiting recursive ($P_S(x)=1$ if $x\in S$, otherwise $P_S(x)=0$) is 2-recursive and limiting recursively enumerable functions ($P_S(x)$ converges iff $x\in S$) is 2-r.e. The term 2-recursive means that the functions are in EA and AE in the arithmetical hierarchy, and 2-r.e. is EA. Stephan and Zeugmann showed that although problems like the Halting Problem and other function classes are learnable in the limit, some ae not reliably learned. Terwijn extends these results by considering learning algorithms where zero, one or two mind changes are allowed and compares them to a function with access to an oracle for the set to be learned.

Just as in the class of computatble functions, there are functions with propertiers that are similar that fall into this class. We will explore some properties that are Computable in the Limit and some that are not.

\section{Definitions}

Using the notion of Rogers\cite{rogers1967} we define the Tau function as $$\tau(x,y)=\frac{1}{2}(x^2+2xy+y^2+3x+y)$$
Where $$\tau^1(x)=x$$ $$\tau^2(x,y)=\tau(x,y)$$ $$\tau^{k+1}(x_1,...,x_{k+1})=\tau(\tau^k(x_1,...,x_k),x_{k+1})$$
and $\pi_i^k(x)$ is the inverse function for the $i^{th}$ element of $\tau^k$.

Without loss of generality, an enumeration of the partial recursive functions is given as $\phi_i$ and the primitive recursive functions is given as $F_i$. A function with two arguments $\phi_i(x,i)$ is simply $\phi_i(\tau(x,y))$ and similarly for $F_i(x,y)$.

We begin with the definition of Computing in the Limit. This is similar to Gold's definition of a Limiting Recursive Set.

\begin{definition}
A function $f_i(x)$ is {\bf Computable in the Limit}, where there exists a corresponding partial recursive function $\phi_i(x, y)$ such that if $f_i(x)=z$ then there exists a $s$ such that for all $t>s$ if $\phi_i(x,t )$ converges, then $\phi_i(x,t )=z$. If this happens, we say {\bf$f_i$ Converges on $x$ to value $z$ in the Limit, or \emph{l-}converges to $z$ on $x$}.

The function {\bf $f_i$ Diverges in the Limit on x} if for all $s$ where $\phi_i(x, s)=v$ then there exists a $t$ such that $\phi_i(x,t )=w$ and $v\neq w$. If this happens, we say {\bf$f_i$ Diverges on value $x$ in the Limit, or \emph{l-}diverges on $x$}.

If the function $f_i$ \emph{l-}diverges on some input $x$ then $f_i$ is {\bf\emph{l-}Partial} .Otherwise, it is {\bf \emph{l-}Total}.
\end{definition}

For example, the set $<x,y>$ of all shortest algorithmic descriptions of Kolmogorov complexity is expressed by such a function $f_K$. $f_K$ takes the input $x$ and generates $y$ if $y$ is the smallest integer such that $\phi_y(0)=x$. This requires a dovetail function $\phi_Ki(x, t)$ that runs all $s\leq t$ where $\phi_s(0)=x$ in time less than $t$ and for all $r<s$ either $\phi_r(0)=z \neq x$ or $\phi_r(0)$ does not converge in time $t$. If no such $\phi_s(0)=x$ is found, return the value $t$.

\section{The Normal Form Theorem}

The notation $f_i(x)$ for properties Computable in the Limit looks similar to the standard definition of a partial recursive function $\phi_i(s)$, except that partial recursive functions are defined using the primitive recursive functions $\Phi_i(x,t)$. Actually we shall see that the only difference is the terminating condition. This allows us to construct a Normal Form Theorem that is similar to Kleene's Normal Form theorem except that instead of using the mu function $\mu(x)$ to find the first element where a Boolean predicate is true, we use a new function $\lambda(x)$ to find the last of a finite number of guesses.

\begin{definition}
The function $\mu(x)[...x...]$ is the least integer $x$ such that the expression $...x...$ is true when ``$x$'' is interpreted as the integer $x$, if $...x...$ is true at least one point\cite{rogers1967}.
\end{definition}

\begin{definition}
The function $\lambda(x)[...x...]$ is the largest integer $x$ such that the expression $...x...$ is true when ``$x$'' is interpreted as the integer $x$, if $...x...$ is true at a finite number of points.
\end{definition}

Kleene's Normal Form Theorem is given as follows.

\begin{thm}
{\bf Kleene's Normal Form Theorem}: There exists a primitive recursive function $U(z)$ and a primitive recursive predicate $T(e,x,y)$ such that every function $f(x)$ is effectively computable iff $f(x)=\phi_e(x)=U(\mu y T(e,x,y))$.
\end{thm}

The Boolean function $T(e,x,y)$ encodes in the variable $y$ the computation history of Turing machine $TM_e$ on input $x$. The predicate $T(e,x,y)$ returns true if $y$ is a halting sequence. The function $U(z)$ recovers the output from the computation history $y$. We shall assume that if $T(e,x,y)$ is true then for all $z>y$ where $z$ is $y$ with one or more copies of the last tape configuration of $y$ appended to $y$, then $T(e,x,z)$ is true also.

\begin{thm}
{\bf Computing in the Limit Normal Form}: There exists a primitive recursive function $U(z)$ and a primitive recursive predicate $T'(e,x,y)$ such that every function $f(x)$ is Computable in the Limit iff $f(x)=U(\lambda y T'(p,x,y))$.
\end{thm}

\begin{proof}
By the definition of $P_p(x)$, there is a $\phi_p(x,n)$ that computes $P_p$. By definition, $\phi_p(x,n)=\phi_p(\langle x,n\rangle )=TM_p(\tau(x,n))$. By Kleene's Normal Form Theorem there is a $U$ and $T$ such that
$$\phi_p(x,n)=U(\mu y T(p,\tau(x,n),y))$$

Define $T'$ from $T$ as follows (for each $p$ and $x$):

$T'(p,\tau(x,n),y)$ is true iff $T(p,\tau(x,n),y)$ is true and for all $m<n$ and $z<y$ where $T(p,\tau(x,m),z)$ is true, then $U(y)\neq U(z)$.

Note that this condition does not use the $\mu$ function.

Assuming $T(p,x,y)$ and $U(x)$ are primitive recursive, $T'(p,\tau(x,n),y)$ is also.

If the following two conditions are true for some time $n$:
\begin{itemize}
\item Let $\{\langle i_1,y_1\rangle ,...,\langle i_n,y_n\rangle ,...\}$ be the set of all pairs such that for each $n$, $T(p,\tau(x,i_n),y_n)$ is true.
\item For all $m>n$ if $T(p,\tau(x,i_m),y_m)$ is true then $U(y_n)=U(y_m)$
\end{itemize}
then the following must be true:

If $y_i$ is the largest value in the set $\{y_0, ... y_n\}$, then $T(p,\tau(x,y_i),z)$ is true, where $z>y_i$ and $z$ is a configuration with one or more copies of the last tape configuration in $y_i$ appended to $z$.

So $U(z)$ is the last value where $TM_p$ changes its mind on the input sequence. Therefore, the last time $T'(p,x,y)$ is true is exactly this value $y=z$.

So $P_p(x)=U(\lambda y T'(p,x,y))$.
\end{proof}

Gold shows that limiting primitive recursive, limiting total recursive and limiting partial recursive functions are equally as powerful. This normal form theorem shows why this is so. The use of the lambda function with the primitive recursive functions are powerful enough to compute anything in this class. Since partial recursive functions can be computed with primitive recursive functions and the application of the mu function, the substituion of the mu function by the lambda function gives the same expressive power to a primitive recursive function as to a total or partial recursive function.

\section{Some Properties Computable in the Limit}

Here are a number of properties that are Computable in the Limit but not effectively computable. These are all pretty obvious, so I shall just state them.

Note that Computable in the Limit properties can be \emph{l-}partial functions, where for $P_p(x)$ there are either an infinite number of guesses or no guess at all. Some of the examples here are \emph{l-}total $P_p(x)$ functions. In all of these cases, the class of sequences is $A_Z$:

\begin{itemize}
\item The Kolmorogorov set $K(x)=P_K(x)=y$. The description was given in the Definitions as an example. $P_K$ is \emph{l-}total.

\item Kolmogorov incompressible numbers: $P_I(n)=x$, where $x$ is incompressible. Use $P_K(x)$ and discard any values $y$ that are compressible. When that happens, shift the guesses for $P_I(n)$ down. At each value $n$ there comes a time when all of the previous compressible values are found, and $P_I(n)$ is never shifted. $P_I$ is \emph{l-}total.

\item The set of all partial recursive functions $TM_x$ where $P_P(x)=\langle x,y\rangle $ only if $TM_x$ is partial and $y$ is the smallest value where $TM_x(y)$ diverges. The value $P_P(x,y)$ is computed as follows: for all $z\le y$ run $TM_x(z)$ for $y$ steps at most. If there is a value $z$ where $TM_x(z)$ has not converged in $y$ steps or less, output $\langle x,z\rangle $. If they all converge, output $\langle x,y\rangle $. If $TM_x$ is total, then if $\Phi_x(y)<y$ for all but a finite $y$, then $P_P(x,y)=\langle x,y\rangle $ for all but a finite number of cases and does not \emph{l-}converge. Otherwise, the output of $P_P(x,y)$ changes infinitely often as the smaller inputs converge. In this case, $P_P$ is \emph{l-}partial.

\item The set of all partial recursive functions $TM_x$ where for every $n$ there is a unique $x$ such that $P_Q(n)=x$. Use the same trick as above for the \emph{l-}total property of Kolmogorov incompressible numbers. In this case, we must keep track of both the function $TM_a$ and the value $P_P(a)=\langle a,b\rangle $ at which it diverges. Ordering the $\langle a,b\rangle $ pairs, $P_Q(n)$ outputs the value $c$, where $\langle c,d\rangle $ is the $n^{th}$ smallest pair. This takes \emph{l-}partial $P_P$ and turns it into \emph{l-}total $P_Q$.

\item The set of all partial functions with finite domain. Instead of keeping track of $\langle a,b\rangle $ where $TM_a(b)$ diverges, keep track of $\langle a,\langle b_0,...,b_n\rangle \rangle $ where $TM_a(b_i)$ converges, for each $i\leq n$.

\item The minimum index for finite sets - similar to the Kolmorogorov set.

\item The set of all functions $TM_i$ with a single element in its domain.

\item Given an enumeration of indexes of Polynomial functions $A$ and enumeration of indexes of NP functions $B$, the Boolean property over pairs $\langle i,j\rangle $ where $P_{eq}(\langle i,j\rangle )=1$ iff $i\in A$, $j\in B$ and $\phi_i=\phi_j$. Otherwise $P_{eq}(\langle i,j\rangle )=0$. If $i\notin A$ or $j\notin B$, then $P_{eq}(\langle i,j\rangle )$ diverges.

\item The property $P_{exp}(i)=\langle i,e\rangle $ where $\phi_i$ has polynomial complexity with exponent $e$. If $\phi_i$ is not polynomial, then $P_{exp}(i)$ \emph{l-}diverges.
\end{itemize}

The Polynomial - NP property can be generalized to any pairs of recursively enumerable classes of total functions. We will present two types of generalizations. Let the class of total functions $A$ be enumerated by a total function $\phi_a$ whose range is indexes of total functions in the set $A$, and similarly for $B$ and $\phi_b$. A theorem by Blum and Blum on the extrapolation of total recursive functions uses the concept of an h-easy function. They show that each recursively enumerable class of functions is bounded, up to a finite number of exceptions, by the computational complexity of a total function $h$.\cite{blum1975toward}

\begin{definition}
Let $h$ be a total recursive function. A partial recursive function $\phi_i$ is {\bf h-easy} if $\Phi(x)\leq h(x)$ for all but a finite number of integers $x$.
\end{definition}

This definition differs slightly from Blum and Blum's in that $\phi_i$ can diverge in a finite number of cases.

\begin{thm}
For any two total recursive functions $g$ and $h$, $P_{eq}$ is Computable in the Limit, where if $\phi_i$ is g-easy and $\phi_j$ is h-easy then $P_{eq}(\langle i,j\rangle )=1$ if $\phi_i=\phi_j$ and $P_{eq}(\langle i,j\rangle )=0$ if $\phi_i\neq\phi_j$ and $P_{eq}(\langle i,j\rangle )$ diverges if either $\phi_i$ is not g-easy or $\phi_j$ is not h-easy.
\end{thm}

\begin{proof}
The function $\phi_{eq}(\langle i,j\rangle ,t)$ is computed as follows. 
For all $x\leq t$, run $\phi_i(x)$ for at most $g(x)$ steps and run $\phi_j(x)$ for at most $h(x)$ steps. If $y$ is the largest value where $\Phi_i(y)>g(y)$ or $\Phi_j(y)>h(y)$ or both, then output $y+1$ unless $y+1$ was output before at some time. Otherwise, output $0$ if there is a case $z$ where $\phi_i(z)\neq \phi_j(z)$. If none are found, then output $1$.

The output $y+1$ is a guess that $\phi_i$ is g-easy and $\phi_j$ is h-easy where all exceptions are less than or equal to $y$. From that point on, $\phi_{eq}$ guesses $1$ (the two are equal) until an exception is found.
\end{proof}

So the Boolean test of equality for h-easy classes of total recursive functions is Computable in the Limit.

Note, though, that neither the class of Polynomial functions nor the class of NP functions can be defined in this way, because the h-easy function will grow to include all polynomial exponents. This allows for functions whose computation time is not within any exponent. We can, though, enumerate indexes of the polynomial functions of the NP functions, and use them instead of $g$ and $h$.

\begin{thm}
For any two total recursively enumerable sets $A$ and $B$ of total recursive functions, $P_{eq}$ is Computable in the Limit, where if $\phi_i\in A$ and $\phi_j\in B$ then $P_{eq}(\langle i,j\rangle )=1$ if $\phi_i=\phi_j$ and $P_{eq}(\langle i,j\rangle )=0$ if $\phi_i\neq\phi_j$ infinitely often, and $P_{eq}(\langle i,j\rangle )$ diverges if either $\phi_i\notin A$ or $\phi_j\notin B$.
\end{thm}

\begin{proof}
In this case, $\phi_{eq}(\langle i,j\rangle ,t)$ checks that $\phi_i\in A$ and $\phi_j\in B$ before running the test for equality.
\end{proof}

Note that these two theorems do not work for sets of partial recursive functions, because we can get stuck on a case where $\phi_i(x)$ converges and $\phi_j(x)$ diverges. The procedure does not identify this case, so it will return equality even though one function diverged.

It is also possible to Compute in the Limit the ratio of equal to unequal values in the two classes, if the ratio is a rational number bounded in the limit.

\begin{definition} Given two total functions $F$ and $G$, the {\bf Error of G on F} (or $F$ on $G$) is the value
$Err(x,F,G)=\left|\left\{{y\le x | F(y)\ne G(y)}\right\}\right| /x$. This is a rational number in the range $0$ to $1$.
\end{definition}

\begin{definition}
{\bf The Error of G on F Converges to v} if there exists an error bound $\varepsilon(x)=a/b$ such that for each $x$, the value $a/b$ is a rational number in the range $0$ to $1$, where $\varepsilon(x)$ is subject to the following two conditions:

\begin{itemize}
\item For all $u$ and $v$, if $u<v$ then $\varepsilon(u)\geq\varepsilon(v)$.
\item For all rational numbers $c/d$ in the range $0$ to $1$ where $c/d\neq0$ there exists a $t$ where $c/d>\varepsilon(t)$.
\end{itemize}

Then there exists an $m$ such that for all $n\geq m$ if $Err(n,F,G)=w$ then $|v - w|\leq\varepsilon(n)$.
\end{definition}

\begin{thm}\label{thmRational}
For any two recursively enumerable classes of total recursive functions, $A$ and $B$, enumerable by $\phi_a$ and $\phi_b$, and an error bound $\varepsilon(x)$, $$P_{err}(\langle i,j\rangle )=\langle i,j,x,y\rangle $$ is Computable in the Limit, where for some $n$ and $m$, $\phi_a(n)=i$ and $\phi_b(m)=j$ and the error of $\phi_i$ on $\phi_j$ converges to a rational number $x/y$.
\end{thm}

\begin{proof}
The function $\phi_{err'}(\langle i,j\rangle ,t)$ is computed as follows. Run $\phi_i(x)=\phi_j(x)$ for each $x\leq t$. Compute
$Err(t,\phi_i,\phi_j)=w$. Find the rational number $a/b$ with the smallest denominator $b$ such that $|w - a/b| \leq\varepsilon(n)$.

The assumption is that the error of $\phi_i$ on $\phi_j$ converges to a rational number $x/y$. So there is an $m$ such that for all $n\geq m$ if $Err(n,F,G)=w$ then
$|w - x/y|\leq\varepsilon(n)$.

If $e/f$ is a rational number where $e/f\neq x/y$ then there exists a rational $c/d$ where $|e/f - x/y|=c/d\neq 0$. So there exists a $t$ where $c/d>\varepsilon(t)$. For all $s\geq t$,
$|w - e/f| >\varepsilon(n)$ so $e/f$ is not chosen after time $t$. This is true for all rational values $e/f$ where $f<y$ so all rationals in the range $0$ to $1$ are rejected in favor of $x/y$ at some time $r$. Therefore $\phi_{err'}(\langle i,j\rangle ,r)=x/y$ and for all $s>r$, $\phi_{err'}(\langle i,j\rangle ,s)=\phi_{err'}(\langle i,j\rangle ,r)$.
\end{proof}

Note that, although this theorem does not apply to the reals in general, any real number that can be computed as the output of a Turing machine starting with a blank tape (such as $e$ or $\pi$) can be added to this function by using the index of the associated Turing Machine as a possible output.

Given a Computing in the Limit Problem $P_p$ where the output is a class of total functions $\phi_i$, we can generate a Canonical listing, where every function appears only once.

\begin{definition}
A {\bf Canonical Enumeration} of total functions is a Computing in the Limit Problem $P_p$ where every function is unique. That is, for any $x$ and $y$, if $P_p(x)=i$ and $P_p(y)=j$ then there is a $z$ such that $\phi_i(z)\neq\phi_j(z)$.
\end{definition}

\begin{thm}
Assume a class of total functions is Computable In the Limit by $P$, in the sense that for each $x$, if $P(x)=y$, then $\phi_y$ is a total recursive function. Then there is a \emph{l-}total $P'$ where an index of each function is given exactly once.
\end{thm}

\begin{proof}
$P'$ is derived from $P$ where we only add in a function if it differs with each of the previous functions by at least one input. If $P$ changes its mind, we have to recompute the differences. Eventually, each output for $P'$ \emph{l-}converges.
\end{proof}

\begin{definition}
A {\bf Complexity Bound Enumeration} of a canonical enumeration of total functions is an enumeration of all values $i$ where there is a $\phi_j$ in the canonical enumeration such that
\begin{itemize}
\item For all $x$, $\phi_i(x)=\phi_j(x)$
\item For all $x$, $\Phi_i(x)\leq\Phi_j(x)$
\end{itemize}
\end{definition}

A Complexity Bound Enumeration finds all of the algorithms that are faster than the one in the Canonical Enumeration.

\begin{thm}
If a class of total functions is Canonically Enumerable, then the complexity bound enumeration is Computable in the Limit.
\end{thm}

\begin{proof}
Given the canonical listing, diagonalize each canonical index $i$ over the enumeration of all Turing machines $TM_j$, discarding any that differ ($\phi_i(x)\neq \phi_j(x)$) or take too much time ($\Phi_i(x)> \Phi_j(x)$).
\end{proof}

\section{Some Properties That Are Not Computable in the Limit}

We end with a couple of examples of properties that are not Computable in the Limit.

\begin{definition}
A {\bf Complete Enumeration} of a class of functions $A$ is an \emph{l-}total property $P_p$, Computable in the Limit, such that, for all $TM_i\in A$, every function $TM_j=TM_i$ has a value $x$ where $P_p(x)=j$.
\end{definition}

\begin{thm}
The Complete Enumeration of all Total Recursive Functions is not Computable in the Limit by any \emph{l-}Total $P_p$.
\end{thm}

\begin{proof}
Assume the opposite: the complete enumeration of all total recursive functions is Computable in the Limit by a \emph{l-}Total $P_p$.

Since $\phi_p$ is \emph{l-}total then for each $x$ there exists a $u$ such that for all $t\geq u$ if $\phi_p(x,t)$ converges, then $\phi_p(x,t)=\phi_p(x,u)=y$ and $\phi_y$ is a total recursive function. Also, every index $j$ of a total function $TM_j$ is an output of $P_p(v)$ for some $v$.

Define a family of functions $f(i)$ where for each $i$, $\phi_{f(i)}(\langle x,u\rangle )$ is constructed from $\phi_p(x,t)$ as follows: begin by running $\phi_p(x,t)$ for all $t\leq (u+i)$ for $(u+i)$ steps at most. If no such $\phi_p(x,t)$ converges in time $(u+i)$ or less, then dovetail the computations of $\phi_p(x,t)$ for all t, and find the first $v$ where $\phi_p(x,v)$ converges. If $w$ is the largest value where $\phi_p(x,w)$ converges in this computation and $\phi_p(x,w)=y$, then run $\phi_y(\langle x,u\rangle )$. If $\phi_y(\langle x,u\rangle )$ converges, where $\phi_y(\langle x,u\rangle )=z$ then output $z+1$.

By the assumption, for all $x$, $P_p(x)$ \emph{l-}converges to the index of a total recursive function $y$, although the initial guesses may be of partial recursive functions. If $\phi_{f(i)}(\langle x,u\rangle )$ uses one of these guesses, it will never halt, at least on a finite number of initial values of $u$. But there will come a time $s$ where for all $t\geq s$ the function $\phi_{f(t)}(\langle x,u\rangle )$ converges to a value for every $u$ and is therefore total.

Choose $j$, $f(t)=j$ where $\phi_j$ is total. Since $P_p$ is a complete enumeration, there will be an index $y$ such that $P_p(y)=j$ in the limit, so after a finite value $u$, for all $t>u$ $\phi_p(\langle y,t\rangle )=j$ if it converges. Let $\phi_j(\langle y,t\rangle )=z$. This value $z$ exists, since $\phi_j$ is total. But by the construction of $\phi_{f(t)}=\phi_j$ given above, $\phi_{f(t)}(\langle y,s\rangle )=\phi_j(\langle y,s\rangle )=\phi_j(\langle y,s\rangle )+1$ a contradiction. So the Complete Enumeration of all Total Recursive Functions is not Computable in the Limit by a \emph{l-}total $P_p$.
\end{proof}

We can strengthen this result by showing that no Complete Enumeration of even a single total recursive function is Computable in the Limit by any $P_p$, even if it is \emph{l-}partial.

\begin{thm}
Given any total recursive function $TM_i$, the Complete Enumeration of all Total Functions $TM_j$ equal to $TM_i$ is not Computable in the Limit by any $P_p$.
\end{thm}

\begin{proof}
Assume the opposite: for an arbitrary $TM_i$, there is a $P_p$ such that for all $j$, $TM_i=TM_j$ iff there is an $x$ such that $P_p(x)=j$. Let $P_p(x)$ be computed by $\phi_p(x,t)$.

Given $i$ and $\phi_p$, define $TM_n(y)$ as follows. Run all $\phi_p(x,t)$ computations for $x\leq y$ and $t\leq y$ for up to $\Phi_p(x,t)\leq y$. If there is no case where $\phi_p(x,t)=n$ then $TM_n(y)=TM_i(y)$.

Otherwise, enumerate all cases $\langle x,t\rangle $ where $\phi_p(x,t)=n$ for $x\leq y$ and $t\leq y$ in time $\Phi_p(x,t)\leq y$. Let $y=\tau(m,q)$ and select the $m^{th}$ case $\langle z,s\rangle $ where $\phi_p(z,s)=n$. This ensures that each such case gets chosen an infinite number of times during the computations of all inputs $y$ to $TM_n(y)$. Dovetail the computations of all $\phi_p(z,r)$ for $r>s$ until a value $v$ is found where $\phi_p(z,v)=m$ and $m\neq n$. If none is found, then $TM_n(y)$ diverges. Otherwise, $TM_n(y)=TM_i(y)$.

If $P_p(y)$ never equals $n$ for all $y$, then $TM_n=TM_i$. Then $P_p$ cannot be a Complete Enumeration, since it missed $TM_n$.

If $P_p(y)=n$ for some $y$ then there is a value $s$ where $\phi_p(y,s)=n$. At this point $\phi_p(y,t)$ either diverges or $\phi_p(y,t)=n$ for all $t>s$. Let $TM_n(v)$ be a value $v$ where this value $\phi_p(y,s)=n$ is the case selected in the computation of $TM_n(v)$. By the construction of $TM_n$, $TM_n(v)$ diverges, and is therefore not a total function.

Since this construction is true for any $i$ and $\phi_p$, the Complete Enumeration of all Total Functions $TM_j$ equal to $TM_i$ is not Computable in the Limit by any $P_p$.
\end{proof}

\section{Conclusions}

The class of properties that are Computable in the Limit are an interesting extension of the effectively computable functions. Although they are not computable in a finite time, they model our everyday notion of what learning and generalization are. The restricted subset of problems termed ``Identification in the Limit'' have been extensively covered in the Machine Learning and Computational Learning Theory fields. But the formal properties of the class itself, outside of its use as a model for learning, is itself interesting. This paper serves as a start for the exploration of this class.

It is obvious that the substitution of $\mu(x)$ for $\lambda(x)$ in the Normal Form Theorem can be further extended to classes where the computing function changes its mind an infinite number of times, and there is some predicate that defines a sort of limiting condition of that sequence. This would make it possible to extend Theorem \ref{thmRational} to the reals. This extension would of course be another superset, with characteristics all its own.

\section{Acknowledgements}

I would like to acknowledge Patrick Boyle and David Buhanan for helpful discussions.

\bibliographystyle{plain}	
\bibliography{limit}	

\end{document}